\documentclass[twocolumn,pra,aps,amssymb]{revtex4}
\usepackage{algorithm2e}
\usepackage{amsmath,amssymb,amsfonts,amsthm}
\usepackage{multirow}
\usepackage{verbatim}
\usepackage{url}
\usepackage{comment}

\usepackage{xcolor}
\begin{document}
\newcommand{\ket}[1]{\ensuremath{\left|#1\right\rangle}}
\newcommand{\bra}[1]{\ensuremath{\left\langle#1\right|}}
\newtheorem{definition}{Definition} 
\newtheorem{theorem}{Theorem} 
\newtheorem{claim}{Claim}
\title{Secure two-party quantum computation for non-rational and rational settings}
\author{Arpita Maitra$^1$, Goutam Paul$^2$ and Asim K. Pal$^1$}
\affiliation{$^1$Management Information Systems Group,\\
Indian Institute of Management Calcutta, India.\\
Email: \{arpitam, asim\}@iimcal.ac.in\\
$^2$ Cryptology \& Security Research Unit,\\
R. C. Bose Centre for Cryptology \& Security,\\
Indian Statistical Institute, Kolkata,\\
Email: goutam.paul@isical.ac.in
}

\begin{abstract}
Since the negative result of Lo (Physical Review A, 1997), it has been left open whether there exist some functions that can be securely computed in two-party setting in quantum domain when one of the parties is malicious. In this paper, we for the first time, show that there are some functions for which secure
two-party quantum computation is indeed possible for non-simultaneous channel model. This is in sharp contrast with the impossibility result of Ben -Or et al. (FOCS, 2006) in broadcast channel model. The functions we study are of two types - one is any function without an embedded XOR, and the other one is a particular function containing an embedded XOR. Contrary to classical solutions, security against adversaries with unbounded power of computation is achieved by the quantum protocols due to entanglement. Further, in the context of secure multi-party quantum computation, for the first time we introduce rational parties, each of whom tries to maximize its utility by obtaining the function output alone. We adapt our quantum protocols for both the above types of functions in rational setting to achieve fairness and strict Nash equilibrium.
\end{abstract}
\maketitle

\section{Introduction}
In a secure two-party computation, two parties or players want to compute 
a particular function of their inputs keeping the inputs secret from each other. They are only allowed to obtain the output of the function preserving some security notions under certain adversarial model.

The secure two-party computation is a special case of `Secure Multi-party Computation' (SMC). In classical domain, the SMC problem has been 
studied extensively. The security of classical SMC comes from some
computational hardness assumptions and thus is conditional.
On the other hand, in quantum domain the adversary is always assumed to have unbounded power of computation and the security of a protocol comes from the laws of physics. 
This is why many researchers have
tried to exploit the quantum mechanical effect~\cite{S12} to solve the problems of SMC~\cite{CXNWY,He,HMG,JWSG,LWJ,Lo,TLH,YXJZ}.

 In~\cite{Lo}, it is pointed out that there are some functions which can not be securely evaluated in quantum domain for two-sided~\cite{footnote} two party setting. Later, Ben -Or et al.~\cite{Ben} generalized it by showing an impossibility result for $n$ players, when there are $\frac{n}{2}$ or more faulty players.
Since the work of~\cite{Lo} in 1997, in case of two-party quantum computation, some additional assumptions, such as the semi-honest third party etc., have been introduced to obtain the secure private comparison~\cite{CXNWY,TLH,YW}.

Yao's millionaires' problem~\cite{Yao} is one of the examples of the secure 
two-party computation. Yao's millionaires' problem~\cite{Yao}, or more precisely, the `greater than' function deals with two millionaires, Alice and Bob, who are interested in finding who amongst them is richer, without revealing their actual wealth to each other. Much effort has been given to solve this problem in quantum domain~\cite{HMG,JWSG,He,YXJZ,TLH}, all of which analyzed the security issues against several eavesdropping strategies. 
Jia et al.~\cite{JWSG} dealt the problem with semi-honest party. 
In~\cite{HMG}, the millionaires' problem is studied considering continuous variable. 
He~\cite{He} exploited the idea of quantum key distribution to solve the problem. 
Tseng et al.~\cite{TLH} proposed the use of Bell state to solve this problem. Their protocol also exploits a third party to assist the players. Yang et al.~\cite{YXJZ} showed the vulnerability of their protocol if the third party is disloyal. However, none of these works~\cite{CXNWY,HMG,JWSG,He,LWJ,YXJZ,TLH,YW} analyze the security issues considering malicious players.

In classical domain the subsequent work by Gordon et al.~\cite{Gordon,GK11}
showed that any function over polynomial-size domains which does not contain an ``embedded XOR" can be converted into the greater than function or more specifically into  the millionaires' problem. Hence, millionaires' problem covers all functions without embedded XOR. Gordon et al. also studied a function which has an embedded XOR~\cite{Gordon,GK11}, namely, a function that simply checks whether the inputs chosen by two players (from a specified domain) are equal or not. Exploiting the idea of Gordon et al., we for the first time design two quantum protocols for these two distinct sets of functions and analyze the security issues when players are malicious unlike the existing quantum protocols~\cite{CXNWY,HMG,JWSG,He,LWJ,YXJZ,TLH,YW}.

Further, we analyze our new quantum protocols considering rational players and this is the first work on secure multi-party quantum computation in rational setting. Rational players are neither `good' nor `malicious', they are utility maximizing. Each rational party wishes to learn the output while allowing as few others as possible to learn the output. Thus, each rational party chooses to abort to maximize its utility. This rationality concept comes from game theory. 
Recently, significant effort has been given towards bridging the gap between two apparently unrelated domains, namely, cryptography and game theory~\cite{GKTZ,AL11,GroceK}. Cryptography deals with the worst case scenario, making the protocols secure against malicious behaviour of a party. However, in game theoretic perspective, a protocol is designed against the rational deviation of a party.
Very recently, Brunner and Linden~\cite{Brunner} showed a deep link between quantum physics and game theory. By bringing quantum mechanics into a class of games, known as Bayesian games, they showed that players who can use quantum resources, such as entangled quantum particles, can outperform classical players. In quantum domain, the concept of rational players in secret sharing has been first introduced in~\cite{ASPP}. In this paper, we identify that fairness in secure two-party computation in non-rational setting does not imply fairness in rational setting. In rational setting, we modify the protocols to achieve both 
fairness as well as strict Nash equilibrium~\cite{AL11,ASPP}.

\subsection{Contributions} 
Below we summarize our contributions in this work.
\begin{enumerate}
\item  For the first time in quantum domain, we identify that for non-simultaneous channel model, there exist some functions which can be computed in two-party setting with complete fairness when one of the parties acts maliciously. 
We consider two sets of functions. One set consists of the functions without embedded XOR, whereas the other set deals with a specific function having an embedded XOR.
\item We also consider rational adversaries and modify our protocols accordingly to achieve both fairness and strict Nash equilibrium. To our knowledge, ours is the first work on secure multi-party quantum computation in the rational 
setting. 
\item Our protocols are secure against both Byzantine as well as Fail-stop adversaries in both non-rational and rational settings.
\end{enumerate}

\subsection{Key Differences from Prior Works}
Here we highlight the key differences of our protocols from the existing quantum protocols for secure two and multi-party computation.
\begin{enumerate}  
\item  Lo~\cite{Lo} showed that, there are certain functions for which two-sided secure two-party quantum computation is impossible if one of them is malicious. Ben -Or et al.~\cite{Ben} proved that assuming pairwise quantum channels and classical broadcast channels among the $n$ players, a universally composable, statistically secure multi-party quantum computation is possible for less than $\frac{n}{2}$ faulty players. On the other hand, we identify that in non-simultaneous channel model, both the millionaires' and the embedded XOR problem can be solved in quantum domain with complete fairness when one of the 
parties is malicious. 
\item Our protocols differ from the existing quantum protocols for private comparison~\cite{CXNWY,HMG,JWSG,He,LWJ,YXJZ,TLH,YW} in the sense that all these protocols analyze the security issues against several eavesdropping strategies. None of those consider malicious players. Contrary to this, we analyze the security of our protocols considering malicious behaviour of the players. In our protocols there are no external adversary. 
\end{enumerate}

\section{Preliminaries} 
\label{pre}In this section we explain what is meant by functionality, two-party computation, ideal and real world model, security of a protocol, Byzantine and fail-stop adversary used in this work. We also define fairness in non-rational as well as rational settings. We identify that when we move from one model to another, the definition of fairness changes. Further, we define strict Nash equilibrium for two players game in the rational setting.

\subsection{Functionality}
In classical domain and in two-party setting, a functionality $\mathcal{F}=\{f_\lambda\}_{\lambda\in\mathbb{N}}$ is a sequence of randomized processes, where $\lambda$ is the security parameter and $f_\lambda$ maps pairs of inputs to pairs of outputs (one for each party). Explicitly, we can write $f_\lambda=(f^1_{\lambda},f^2_{\lambda})$, where $f^1_{\lambda}$ (resp. $f^2_{\lambda}$) represents the output of the first party, say $P_1$ (resp. output of the second party, say $P_2$). The domain of $f_\lambda$ is $X_\lambda \times Y_\lambda$, where $X_\lambda$ (resp. $Y_\lambda$) denotes the possible inputs of the first (resp. second) party. If the domain sizes $|X_\lambda|$ and $|Y_\lambda|$ are polynomial in $\lambda$, then we say that $\mathcal{F}$ is defined over polynomial size domains. If each $f_\lambda$ is deterministic we say that each $f_\lambda$ as well as the collection $\mathcal{F}$ is a function. 

\subsection{Two-Party Computation}
In classical domain, the two-party computation of a functionality $\mathcal{F}=\{f_\lambda^{1},f_\lambda^{2}\}$ is defined as follows.
If a party $P_1$ is holding $1^{\lambda}$ and a input $x\in{X_\lambda}$ and a party $P_2$ is holding $1^{\lambda}$ and a input $y\in{Y_\lambda}$, then the joint distribution of the outputs of the parties is statistically close to $(f_\lambda^{1}(x,y),f_\lambda^{2}(x,y))$. 

\subsection{Ideal vs. Real World model} 
In {\em ideal world model} we assume that there is an incorruptible trusted third party who computes the function in behaves of $P_1$ and $P_2$. $P_1$ and $P_2$ send their inputs to the TTP who computes the functionality and returns the value to each party. On the other hand, in {\em real world model} there is no trusted party to compute the functionality, rather a protocol is executed to compute the functionality.

Here, along the same line as~\cite{Gordon,GK11}, we assume a {\em hybrid world model}, where there is a trusted third party who computes the function like in the ideal world and distributes the shares of the function's output like a dealer in secret sharing~\cite{Shamir} between the players. The players construct the output by exchanging their shares. In our hybrid world model we call the TTP as a {\em dealer}.

The security of a protocol depends upon what an adversary can do during the real protocol execution. In ideal world, as there is an incorruptible trusted third party who computes the function and sends the output to the participants the computation is secure by definition. However, in real world model there is no trusted party. If the adversary who exists in the real model can do no more harm than the ideal scenario, then we say that the protocol is secure. 

\subsection{Fail-stop and Byzantine Adversarial Model}
In the fail-stop setting, each party follows the protocol as directed except that it may choose to abort at any time~\cite{GroceK} and a party is assumed not to change its input when running the protocol. On the other hand, in Byzantine setting, a deviating party may behave arbitrarily. It may change the inputs or may choose to abort. Since Byzantine adversary covers all the characteristics of a fail-stop  adversary, it is very natural to consider only Byzantine setting. If a protocol is secure against a Byzantine adversary, it must be secure against a fail-stop adversary. Hence, throughout the paper we analyze the security issues against Byzantine adversary only. 

\subsection{Security in Non-rational Setting}
In non-rational setting, the move of a player is decided by his adversarial nature not by his utility function; whereas in rational setting every move of a player is guided by his utility. 

\subsubsection{\bf {Fairness}} 
For fairness in non-rational setting, we need to introduce some terminologies. Let us assume that $P_1$ begins by holding an input $x\in X$ and $P_2$ begins by holding an input $y\in Y$, and $z \in \{0,1\}^*$ is the auxiliary input of the adversary. Let $\{IDEAL_{\mathcal{F},\mathcal{S}(z)}(x,y)\}_{(x,y)\in X\times Y, z\in\{0,1\}^*}$ represent a pair of two random variables denoted by $VIEW$ and $OUT$, where $VIEW_{ideal}(x,y)$ represents the output of the party who is corrupted by the adversary $\mathcal{S}$ and $OUT_{ideal}(x,y)$ represents the output of the honest party in the ideal world. Thus, we can write 
\begin{eqnarray*}
\{IDEAL_{\mathcal{F},\mathcal{S}(z)}(x,y)\}_{(x,y)\in X\times Y, z\in\{0,1\}^*}\\
 = (VIEW_{ideal}(x,y), OUT_{ideal}(x,y)).
\end{eqnarray*}
Similarly, let $\{REAL_{ \bf\Pi,\mathcal{A}(z)}(x,y)\}_{(x,y)\in X\times Y, z\in\{0,1\}^*}$ represents a pair of two random variables, namely $VIEW_{real}(x,y)$ and $OUT_{real}(x,y)$, where $VIEW_{real}(x,y)$ denotes the random variable in real world consisting of the view of the player corrupted by the adversary $\mathcal{A}$ and $OUT_{real}(x,y)$ represents the random variable consisting of the  output of the honest party in the real world~\cite{lindell}. 
\begin{definition}{(Fairness)} A protocol ${\bf\Pi}$ is said to securely compute a functionality $\mathcal {F}$ with complete fairness if for every adversary $\mathcal {A}$,  having unbounded power of computation in the real model, there exits an adversary, $\mathcal{S}$, with same computational complexity in the ideal model such that
\begin{eqnarray*}
& & \{IDEAL_{\mathcal{F},\mathcal{S}(z)}(x,y)\}_{(x,y)\in X\times Y, z\in\{0,1\}^*}\\
& {\equiv} & \{REAL_{ \bf\Pi,\mathcal{A}(z)}(x,y)\}_{(x,y)\in X\times Y, z\in\{0,1\}^*}.
\end{eqnarray*}
\end{definition}
Note that, here we do not require a security parameter $\lambda$ as we consider our adversary has unbounded power of computation. 

In our hybrid model, the fairness condition is as follows.
\begin{definition}{(Fairness)} A protocol ${\bf\Pi}$ is said to securely compute  a functionality $\mathcal {F}$ with complete fairness if for every adversary $\mathcal {A}$,  having unbounded power of computation in the hybrid model, there exits an adversary, $\mathcal{S}$, with same computational complexity in the ideal model such that
\begin{eqnarray*}
& & \{IDEAL_{\mathcal{F},\mathcal{S}(z)}(x,y)\}_{(x,y)\in X\times Y, z\in\{0,1\}^*}\\
& {\equiv} & \{HYBRID_{ \bf\Pi,\mathcal{A}(z)}(x,y)\}_{(x,y)\in X\times Y, z\in\{0,1\}^*}.
\end{eqnarray*}
\end{definition}
here, $REAL$ is replaced by $HYBRID$ which is the random variable  consisting of the view ($VIEW$) of the adversary and the output ($OUT$) of the honest party in the hybrid world in the same manner as above.

\subsection{Rational Setting and its Security}
 We define a {\em function reconstruction protocol with rational players}
to be a pair $(\Gamma, \overrightarrow{\sigma})$, where $\Gamma$ is the game (i.e., specification of allowable actions) and $\overrightarrow{\sigma}$=$(\sigma_1,\ldots,\sigma_n)$ denotes
the strategies followed by $n$ number of players. We use the notations
$\overrightarrow{\sigma}_{-w}$ and $(\sigma'_w,\overrightarrow{\sigma}_{-w})$
respectively
for $(\sigma_1,\ldots,\sigma_{w-1},\sigma_{w+1},\ldots,\sigma_n)$ and $(\sigma_1,\ldots,\sigma_{w-1},\sigma'_w,\\
\sigma_{w+1},\ldots,\sigma_n)$.
The outcome of the game is denoted by $\overrightarrow{o}(\Gamma, \overrightarrow{\sigma})$=$(o_1,\ldots,o_n)$. The set of possible outcomes with
respect to a party $P_w$ is as follows.
1) $P_w$ correctly computes $f$, while others do not; 2) everybody correctly
computes $f$; 3) nobody computes $f$; 4) others computes $f$ correctly, 
while $P_w$ does not.

The output that no function is computed is denoted by $\perp$ (i.e., {\em null}
as in~\cite{Gordon}).

\subsubsection{{\bf Utilities and Preferences}}
The utility function $u_w$ of each party $P_w$ is defined over the set of 
possible outcomes of the game. The outcomes and corresponding utilities for two
parties are described in Table~\ref{table: OutcomesRSS}. We here assume that the utility values are real.  

\begin{table}[htbp]
\caption{Outcomes and Utilities for $(2,2)$ rational function reconstruction}
\label{table: OutcomesRSS}
\begin{center}
\begin{tabular}{llll}
\hline\noalign{\smallskip}
$P_1$\rq s outcome & $P_2$\rq s outcome & $P_1$\rq s Utility & $P_2$\rq s Utility\\
$(o_1)$ & $(o_2)$ & $U_1(o_1, o_2)$ & $U_2(o_1, o_2)$\\
\hline
\noalign{\smallskip}
$o_1$=$f$ & $o_2$=$f$ & $U_1^{TT}$ & $U_2^{TT}$\\
$o_1$=$\perp$ & $o_2$=$\perp$ & $U_1^{NN}$ & $U_2^{NN}$\\
$o_1$=$f$ & $o_2$=$\perp$ & $U_1^{TN}$ & $U_2^{NT}$\\
$o_1$=$\perp$ & $o_2$=$f$ & $U_1^{NT}$ & $U_2^{TN}$\\
\hline
\end{tabular}
\end{center} 
\end{table}

Players have their preferences based on different possible outcomes. In this work, a rational player $w$ is assumed to have the following preference:
$$\mathcal{R}_1 : U_w^{TN}> U_w^{TT}>U_w^{NN}>U_w^{NT}.$$ 

\subsubsection{\bf {Fairness}}\label{fair} In non-rational setting, the security of a protocol is analyzed by comparing 
what an adversary can do in a real protocol execution to what it can do in an ideal scenario that is secure by definition~\cite{lindell,Gordon,GK11}. This is formalized by considering an ideal computation involving an incorruptible trusted party to whom the parties send their inputs. The trusted party computes the functionality on the inputs and returns to each party its respective output. Loosely speaking, a protocol is secure
if any adversary interacting in the real protocol (where no trusted party exists) can do no more harm than if it were involved in the above-described ideal computation.

A rational player, being selfish, desires an unfair outcome, i.e., computing the function alone. Therefore, the basic aim of rational computation has been to achieve fairness. 
According to Von Neumann and Morgenstern {\em expected utility theorem}~\cite{VM44}, under natural assumptions, the individual would prefer one  prospect $\mathcal{O}_1$ over another prospect $\mathcal{O}_2$ if and only if $E[U(\mathcal{O}_1) \geq E[U(\mathcal{O}_2)]$.
The work~\cite{AL10} implicitly uses the expected utility theorem to derive its
results. We also use the same approach and accordingly redefine fairness
as follows.
\begin{definition}   
(Fairness) A rational function reconstruction mechanism $(\Gamma,\overrightarrow{\sigma})$ is said to be completely fair if  for a party $P_w$, $(w\in{\{1,2\}})$, who is corrupted by an adversary having unbounded power of computation, the following holds:
$$
U_w^{TT} \geq E[U_w (\mathcal{O}_l)], 
$$
where $\mathcal{O}_l = \{o_{w}^{1}, \ldots, o_{w}^{n'};  p_1, \ldots, p_{n'}\}$ is any prospect when the player deviates from the suggested strategy and $n'$ is the number of possible outcomes. 
\end{definition}

\subsubsection{ \bf {Strict Nash Equilibrium}}
Now, we define Nash equilibrium for two players game. A suggested strategy $\overrightarrow{\sigma}$ of a mechanism $(\Gamma,\overrightarrow{\sigma})$ is said to be in Nash equilibrium when there is no incentive for a player $P_w$, $w \in \{1,2\}$ to deviate from the suggested strategy, given that other player is following its suggested strategy. There are many variants of Nash equilibrium in game theory literature~\cite{AL11}. However, in the quantum domain, the players are assumed to have unbounded computational power and hence the relevant equilibrium is the strict Nash equilibrium~\cite{AL11,ASPP}. We recall its definition below.
\begin{definition}
(Strict Nash equilibrium) The suggested strategy $\overrightarrow{\sigma}$ in the mechanism $(\Gamma,\overrightarrow{\sigma})$ is a strict Nash equilibrium, if for every player $P_w$, $w \in \{1,2\}$, who possesses unbounded power of computation and for any strategy $\sigma'_w$ which deviates from the suggested strategy $\overrightarrow{\sigma}$, we have $u_w (\sigma'_w,\overrightarrow{\sigma}_{-w} ) < u_w (\overrightarrow{\sigma})$.
\end{definition}

\section{Revisiting the Millionaires' Problem~\cite{Gordon}}
\label{Byz}
In this section, we first describe the millionaires' problem or more precisely, the greater than function, proposed by Gordon et al.~\cite{Gordon,GK11}. 
Let us denote two players by $P_1$ and $P_2$. As we deal with hybrid model, there is a trusted party whom we call dealer. 
Suppose $P_1$ has the secret $i$ and $P_2$ has the secret $j$, 
$1 \leq i \leq M$, $1 \leq j \leq M$, where $M$ is an integer. The dealer gives an ordered list
$X = \{x_1, x_2, \ldots, x_M\}$ to $P_1$ and another ordered list
$Y = \{y_1, y_2, \ldots, y_M\}$ 
to $P_2$. Then $P_1$ sends $x_i$ to the dealer and $P_2$ sends $y_j$ to
the dealer. Let $f$ be a deterministic function which maps 
$X \times Y \rightarrow \{0,1\} \times\{0,1\}$. The function
$f(x_i,y_j)$ can be defined as a pair of outputs, i.e., 
$f (x_i,y_j) = (f_1(x_i,y_j), f_2(x_i,y_j))$, 
where $f_1(x_i,y_j)$ is the output of the first party $P_1$ and $f_2(x_i,y_j)$ is 
the output of the second party $P_2$.
For millionaires' problem, the function is defined as follows~\cite{Gordon,GK11}.
For $w = 1, 2$,
\begin{equation}
\label{eqfw}
 f_w(x_i, y_j) =
  \begin{cases}
   1 & \text{if } i > j; \\
   0 & \text{if } i \leq j.
  \end{cases}
\end{equation}
The protocol proceeds in a series of $M$ iterations. 
The dealer creates two sequences $\{a_l\}$ and $\{b_l\}$,
$l = 1, 2, \ldots, M$, as follows.
$$a_i = b_j = f_1(x_i, y_j) = f_2(x_i, y_j).$$
For $l \neq i$, $a_l = \perp$ and for $l \neq j$, $b_l = \perp$.

Next, the dealer splits the secret $a_l$ into the shares $a^1_l$ and $a^2_l$,
and the secret $b_l$ into the shares $b^1_l$ and $b^2_l$, so that
$a_l = a^1_l \oplus a^2_l$ and 
$b_l = b^1_l \oplus b^2_l$, 
and gives
the shares $\{(a^1_l$, $b^1_l)\}$ to $P_1$ and
the shares $\{(a^2_l$, $b^2_l)\}$ to $P_2$.
In each round $l$, $P_2$ sends $a^2_l$ to $P_1$, who, in turn
sends $b^1_l$ to $P_2$. $P_1$ learns the output value $f_1(x_i, y_j)$ in iteration $i$, and $P_2$  learns the output value $f_2(x_i,y_j)$ in iteration $j$. In a round $l \neq i$ $P_1$ outputs $\perp$ and in a round $l \neq j$ $P_2$ output $\perp$.
As we require three elements, $0$, $1$ and $\perp$, we define $0$ by $00$, $1$ by $11$ and $\perp$ by $01$. Note that the dealer who will distribute the shares is honest and can compute the function described in Equation~\eqref{eqfw}.

The algorithms in the Byzantine setting is same as the fail-stop setting except some additional steps. In Byzantine setting, the shares are signed by the dealer. As explained in~\cite{Gordon,GK11}, exploiting the MAC signature, we can resist the players to send a false share. 

\section{Quantum Solution of Millionaires' Problem in Non-Rational Setting}
\label{nonrational}
In this section, we propose a quantum version of millionaires' problem. It is the quantum analogue of the protocol of Gordon et al. in classical domain~\cite{Gordon,GK11}. However, their security proof is based on some computational hardness in classical domain. Whereas we exploit the property of entanglement to provide security of the protocol in the quantum domain.

Here, we exploit four Bell state basis~\cite{nc04}. 
The maximally entangled two particle state is $\ket{g_0} = \frac{1}{\sqrt{2}}\Big [\ket {0}_1\ket{0}_2 + \ket{1}_1\ket{1}_2\Big]$. This state is called Einstein, Podolsky,  Rosen Pair, in short $EPR$ pair or Bell  state. There are four independent Bell states. They are
{\scriptsize
\begin{eqnarray*}
\ket{g_0} = \frac{1}{\sqrt{2}}\Big[\ket{0}_1\ket{0}_2 + \ket{1}_1\ket{1}_2\Big],
\ket{g_1} = \frac{1}{\sqrt{2}}\Big[\ket{0}_1\ket{0}_2 - \ket{1}_1\ket{1}_2\Big],\\
\ket{g_2} = \frac{1}{\sqrt{2}}\Big[\ket{0}_1\ket{1}_2 + \ket{1}_1\ket{0}_2\Big],
\ket{g_3} = \frac{1}{\sqrt{2}}\Big[\ket{0}_1\ket{1}_2 - \ket{1}_1\ket{0}_2\Big].
\end{eqnarray*}
}
Here, subscript $1$ stands for $P_1$'s qubit and subscript $2$ stands for $P_2$'s qubit.
We need any three of these orthogonal states. In this work, without
loss of generality, we consider $\ket{g_0}$, $\ket{g_1}$ and $\ket{g_2}$.

 Like classical case, the secret of $P_1$ is $i$ and the secret of $P_2$ is $j$, $1\leq i\leq m$, $1\leq j\leq m$ where $m$ is an integer. They want to know whether $i>j$ or $i \leq j$. The dealer supplies them two ordered lists, 
$X = \{x_1, x_2, \ldots, x_m\}$ to $P_1$ and 
$Y = \{y_1, y_2, \ldots, y_m\}$  to $P_2$. $P_1$ chooses $x_i$ and $P_2$ chooses $y_j$ from their respective lists and send those to the dealer. Dealer will compute the function and will distribute the shares (here, qubits) in such a way that $P_1$ will get the value of the function i.e $f_1(x_i,y_j)$ in iteration $i$ and $P_2$ will get the value of the function i.e $f_2(x_i,y_j)$ in iteration $j$. The protocol proceeds in a series of $m$ iteration. In a round $l\neq i$ $P_1$ outputs $\perp$ and in a round $l\neq j$ $P_2$ outputs $\perp$.  The Quantum solution of the millionaires' problem in non-rational setting, is described in Algorithm~\ref{algo_qsharegen1} ($QShareGen$) 
and Algorithm~\ref{algo_fairqcomp1} ($\bf \Pi^{QMP}_{Fair}$).

\restylealgo{boxed}
\begin{algorithm}[htbp]
{\scriptsize
\textbf {Inputs:}\\
 The inputs of the $QShareGen$ are $x_ i$ from $P_1$ and $y_j$ from $P_2$. If one of the received inputs is not in the correct domain, then both the parties are given $\perp$.\\
\textbf{Computation:}\\
Dealer does the following:\\
\begin{enumerate}
\item \hspace*{2mm} (a) If $f_1(x_i, y_j) = f_2(x_i,y_j)= 0$, prepares two copies of $\ket{g_0}=\frac{1}{\sqrt{2}}(\ket {0}_1\ket{0}_2 + \ket {1}_1\ket{1}_2)$. We denote them as $\ket{g'_0}$ and $\ket{g''_0}$. \\
\hspace*{2mm} (b) If $f_1(x_i, y_j) = f_2(x_i,y_j) = 1$, prepares two copies of $\ket{g_1}=\frac{1}{\sqrt{2}}(\ket {0}_1\ket{0}_2 - \ket {1}\ket{1}_2)$. We denote them as $\ket{g'_1}$ and $\ket{g''_1}$.\\ 
\item For each $l \in \{1, \dots, m\}, l \neq i$ and $l \neq j$, prepares two copies of  $\ket{g_2}=\frac{1}{\sqrt{2}}(\ket {0}_1\ket{1}_2 + \ket {1}_1\ket{0}_2)$.\\
\item For $l = i$, prepares one copy of  $\ket{g_2}=\frac{1}{\sqrt{2}}(\ket {1}_1\ket{0}_2 + \ket {0}_1\ket{1}_2)$. We call that $\ket{g'_2}$.\\
\item For $l=j$, prepares one copy of  $\ket{g_2}=\frac{1}{\sqrt{2}}(\ket {1}_1\ket{0}_2 + \ket {0}_1\ket{1}_2)$. We call that $\ket{g''_2}$.\\
\end{enumerate}
\textbf{Output:}\\
\begin{enumerate}
\item For  $l \in \{1, 2, \dots, m\}$ dealer 
prepares a list $list_w$ of shares for each party $P_w$, where $w \in \{1,2\}$ such that for each round each player is given two qubits, marked as $1$st and $2$nd, from two different entangled states.\\
\hspace*{2mm}(a) when $l=i$, $P_1$ is given the first half from $\ket{g'_0}$ or $\ket{g'_1}$ depending on the value of $f_1(x_i,y_j)$ and the first half from the entangled state $\ket{g'_2}$. $P_2$ is given the other halves. For each party, the qubit from $\ket{g'_0}$ or $\ket{g'_1}$ is marked as $1$st qubit for that round and the qubit from $\ket{g'_2}$ is marked as $2$nd qubit for that round.\\
\hspace*{2mm}(b) when $l=j$, $P_2$ is given the second half from $\ket{g''_0}$ or $\ket{g''_1}$ depending on the value of $f_2(x_i,y_j)$ and the second half from the entangled state $\ket{g''_2}$. $P_1$ is given the first halves. For each party, the qubit from $\ket{g''_0}$ or $\ket{g''_1}$ is marked as $2$nd qubit for that round  and the qubit from $\ket{g''_2}$ is marked as $1$st qubit for that round. \\
\hspace*{2mm}(c) for all other rounds, $P_1$ is given the first halves from two different $\ket{g_2}$ states, whereas $P_2$ is given the other halves from the same entangled states. For each party the qubits are marked such a way that the $1$st (resp. $2$nd) qubit of $P_1$ is correlated with the $1$st (resp. $2$nd) qubit of $P_2$.\\  
\hspace*{2mm}(d) each list contains $2m$ number of qubits.\\
\end{enumerate}
}
\caption{$QShareGen$}
\label{algo_qsharegen1}
\end{algorithm}

\restylealgo{boxed}
\begin{algorithm}[htbp]
{\scriptsize
\textbf{Inputs:}\\
 Each of $P_1$ and $P_2$ receives his corresponding list of shares.\\
\textbf{Computation:}\\
The players do the following.\\
\begin{enumerate}
\item Each round is subdivided into two sub-rounds.\\
\item In first sub-round, $P_2$ sends the first qubit of its list for that round to $P_1$. \\
\item In second sub-round, $P_1$ sends the second qubit of its list for that round to $P_2$. \\
\item After receiving the qubits from $P_2$, $P_1$ measures the two qubits in Bell basis.\\
\hspace*{2mm}(a) If $l \neq i$ and the measurement result is $\ket{g_0}$ or $\ket{g_1}$ or $\ket{g_3}$, aborts the protocol and reports forgery by $P_2$. If it is $\ket{g_2}$, concludes $\perp$. \\
\hspace*{2mm}(b) If $l=i$ and the measurement result is $\ket{g_2}$ or $\ket{g_3}$, then aborts the protocol and reports forgery by $P_2$. If the measurement result is $\ket{g_0}$, concludes $f_1(x_i,y_j) = 0$. If it is $\ket{g_1}$, concludes $f_1(x_i,y_j) = 1$. \\ 
\item  After receiving the qubits from $P_1$, $P_2$ measures the two qubits in Bell basis.\\
\hspace*{2mm}(a) If $l \neq j$ and the measurement result is $\ket{g_0}$ or $\ket{g_1}$ or $\ket{g_3}$, aborts the protocol and reports forgery by $P_1$. f it is $\ket{g_2}$, concludes $\perp$. \\
\hspace*{2mm}(b) If $l=j$ and the measurement result is $\ket{g_2}$ or $\ket{g_3}$, then aborts the protocol and reports forgery by $P_1$. If the measurement result is $\ket{g_0}$, concludes $f_2(x_i,y_j) = 0$. If it is $\ket{g_1}$, concludes $f_2(x_i,y_j) = 1$. \\ 
\end{enumerate}
\ \\
\textbf{Output:}\\
\begin{enumerate}
\item $P_1$ obtains its output value i.e either $0$ or $1$ depending upon $f_1(x_i,y_j)$ in iteration $i$ whereas $P_2$ obtains its output value i.e either $0$ or $1$ depending upon $f_2(x_i,y_j)$ in iteration $j$.\\
\item If $P_2$ aborts in round $l$, i.e., does not send its share at that round and $l \leq i$, $P_1$ outputs $1$. If $l > i$, $P_1$ has already determined the output in iteration $i$. Thus it outputs that value. \\
\item  If $P_1$ aborts in round $l$, i.e., does not send its share at that round and $l \leq j$, $P_2$ outputs $0$. If $l > j$, $P_2$ has already determined the output in iteration $j$. Thus it outputs that value. \\
\end{enumerate}
}
\caption{$\bf\Pi^{QMP}_{Fair}$}
\label{algo_fairqcomp1}
\end{algorithm}

\subsection{Security Analysis}
A Byzantine player can behave arbitrarily. He can manipulate the shares (here, qubits) which he has obtained from the dealer or may abort early. In this subsection we will show how entanglement provides the security against such manipulation. The aborting case will be discussed next.

\subsubsection{\bf{Security against Forgery}}
\label{Byz_non} 
Without loss of generality, let us assume that $P_1$ tries to manipulate the qubits obtained from the dealer in the motivation to convey the wrong message to $P_2$. Here, manipulation means sending arbitrary qubit or swapping the qubits of his list. This forgery is detected with significant probability. Here, we assume that $P_1$ sends an arbitrary qubit to $P_2$ in a round $l$. The analysis will be same if we consider the swapping of the qubits of his list.

Like classical MAC signature, in quantum domain, entanglement provides security against such forgery.
According to the protocol, in round $l \neq j$, if no cheating occurs, then $P_2$ will get $\ket{g_2} = \frac{1}{\sqrt{2}}\Big[\ket{0}_1\ket{1}_2 + \ket{1}_1\ket{0}_2\Big].$ In terms of density matrix it can be written as $$\rho=\frac{1}{2}\Big(\ket{0}_1\ket{1}_2 + \ket{1}_1\ket{0}_2\Big)\Big(\bra{0}_1\bra{1}_2 + \bra{1}_1\bra{0}_2\Big).$$ Now, let us assume that $P_1$ sends an arbitrary qubit which is $\ket{\phi}=[\alpha\ket{0}_3 + \beta\ket{1}_3]$, instead of the correct one. In terms of density matrix, the arbitrary state can be written as
\begin{eqnarray*}
\rho_3 = \ket{\phi}\bra{\phi}
& = & \Big[|\alpha|^2\ket{0}_3\bra{0} + \alpha^*\beta\ket{1}_3\bra{0}+\alpha\beta^*\ket{0}_3\bra{1}\\
& & + |\beta|^2\ket{1}_3\bra{1}\Big].
\end{eqnarray*}
 Thus, the state at the end of $P_2$ would be
\begin{eqnarray*}
\rho_{2} & = & [tr_{P_1}(\rho)](\rho_{3}) \\
& = & \frac{1}{{2}}\Big[\ket{1}_2\bra{1}\Big(|\alpha|^2\ket{0}_3\bra{0} +
\alpha^*\beta\ket{1}_3\bra{0} + \alpha\beta^*\ket{0}_3\bra{1}\\
&  & + |\beta|^2\ket{1}_3\bra{1}\Big)+ \ket{0}_2\bra{0}\Big(|\alpha|^2\ket{0}_3\bra{0}  + \alpha^*\beta\ket{1}_3\bra{0}\\
& & + \alpha\beta^*\ket{0}_3\bra{1}+|\beta|^2\ket{1}_3\bra{1}\Big)\Big].\\ 
\end{eqnarray*}

In this case, when $P_2$ will measure qubit $2$ and  qubit $3$ in Bell basis, after measurement, $P_2$ will get either $\ket{g_0}$ or $\ket{g_1}$ or $\ket{g_2}$ or $\ket{g_3}$ with probability $\frac{1}{4}$ instead of $\ket{g_2}$ only. 
The detailed calculations are given here. For the rest of the paper, we will refer this section.

Let us assume that $P_2$ obtains $\ket{g_0}$ after measurement. Thus, the probability that $P_2$ obtains $\ket{g_0}$ is given by 
\begin{eqnarray*}
&& \ket{g_0}_{23}\bra{g_0}(\rho_{2})
= \bra{g_0}\rho_{2}\ket{g_0}_{23}\\
& = & \frac{1}{4}\Big[\Big(\bra{00}_{23}+\bra{11}_{23}\Big)
\Big[\ket{1}_2\bra{1}\Big(|\alpha|^2\ket{0}_3\bra{0}\\
& & + \alpha^*\beta\ket{1}_3\bra{0}+\alpha\beta^*\ket{0}_3\bra{1}
+|\beta|^2\ket{1}_3\bra{1}\Big) \\
& & + \ket{0}_2\bra{0}\Big(|\alpha|^2\ket{0}_3\bra{0} + \alpha^*\beta\ket{1}_3\bra{0}\\
& & +\alpha\beta^*\ket{0}_3\bra{1} +|\beta|^2\ket{1}_3\bra{1}\Big)\Big]
\Big(\ket{00}_{23}+\ket{11}_{23}\Big)\Big]\\
& = & \frac{1}{4}\Big[|\alpha|^2+|\beta|^2\Big]
 = \frac{1}{4}.
\end{eqnarray*}
If $l \neq j$, according to our protocol, $P_2$ should get $\ket{g_2}$ only. But as $P_1$ sends an arbitrary qubit to $P_2$, when measured, $P_2$ gets any one of the four Bell states with probability $\frac{1}{4}$. Thus, if $l \neq j$ and $P_2$ gets $\ket{g_0}$ or $\ket{g_1}$ or $\ket{g_3}$, he immediately concludes that $P_1$ is cheating. The success probability of detecting such cheating for a round $l \neq j$ is $\frac{3}{4}$. 

Similarly, when $l=j$, if $P_1$ does not cheat, $P_2$ would get either $\ket{g_0}$ or $\ket{g_1}$ depending on the value of $f_2(x_i,y_j)$. However, if $P_1$ cheats, when measured, $P_2$ will get any one Bell state. In case of $\ket{g_0}$ and $\ket{g_1}$, he can not detect the cheating because he does not know the value of $f_2(x_i,y_j)$ a priori. However, if he gets $\ket{g_2}$ or $\ket{g_3}$, he immediately detects the cheating with certainty. Thus, the success probability of detecting the cheating when $l=j$ is $\frac{1}{2}$. As, $P_1$ have no idea about the value of $j$, the average success probability of detecting such cheating is
\begin{eqnarray*}
\frac{3}{4}\Pr(l \neq j)+\frac{1}{2}\Pr(l=j) 
 =\frac{3}{4}[\Pr(l<j)+\Pr(l>j)]\\
 + \frac{1}{2}\Pr(l=j).
\end{eqnarray*}

We do not bother about $ \Pr(l>j)$ because, $P_2$ should have no incentive to detect the cheating when $l>j$, as he has already got his output value in round $j$. Thus the total success probability of $P_2$  to detect such cheating is
\begin{equation*}
\frac{3}{4}\Pr(l<j) +  \frac{1}{2}\Pr(l=j)\\
= \frac{3}{4}.\frac{j-1}{m} + \frac{1}{2}.\frac{1}{m}
= \frac{3j-1}{4m}
\end{equation*}

\begin{theorem}
In non-rational setting, the success probability of $P_2$ to detect cheating by $P_1$ who is corrupted by a Byzantine adversary in an arbitrary round $l$  is $\frac{3j-1}{4m}$.
\end{theorem}
Same conclusion can be drawn when we assume $P_2$ is corrupted. In this case, we modify the theorem in the following way.
\begin{theorem}
In non-rational setting, the success probability of $P_1$ to detect cheating by $P_2$ who is corrupted by a Byzantine adversary in an arbitrary round $l$ is $\frac{3i-1}{m}.$
\end{theorem}

\subsubsection{\bf{Fairness against Early Abort}} 
As $P_1$ is always computing its output first followed by $P_2$, the aborting of $P_1$ plays an important role to achieve the fairness of the protocol. The early abort of $P_2$ will terminate the protocol up to that round in which $P_2$ aborts. In that case, either both get the output or none gets the output. Thus, early abort of $P_2$ does not affect the fairness condition. We now concentrate on the early abort of $P_1$.

Let us assume that $P_1$ aborts in round $l$. There are two cases: $i \leq j$ and $i > j$. We analyze each case one by one.\\
\noindent {\bf Case 1: $i \leq j$}.\\
\noindent {\bf Subcase 1(a):} $l<i$. In this case, $P_1$ outputs $\perp$ and $P_2$ outputs $0$. In {\em ideal world model}, the trusted party sends $f_1(x_i,y_j)$ to $P_1$ in iteration $i$ and $f_2(x_i,y_j)$ to $P_2$ in iteration $j$. In all other rounds trusted party sends $\perp$ to both $P_1$ and $P_2$. If a party (say $P_1$) aborts the protocol in an arbitrary round $l$ after getting the output, the trusted party sends the honest party (here, $P_2$) the value of $f_2(x_l,y_j)$. Thus when $P_1$ aborts in round $l<i$, $P_1$ outputs $\perp$ whereas $P_2$ outputs $f_2(x_l,y_j)$. As $i \leq j$ and $l<i$, then $l<j$. So $f_2(x_l,y_j)= 0$ (refer to Equation~\ref{eqfw}). Hence, 

{\scriptsize
$\Pr \Big[\Big(VIEW_{ideal}(x,y), OUT_{ideal}(x,y)\Big) 
= (\perp, 0)\Big|l<i \bigwedge i \leq j \Big]$

\noindent $= \Pr \Big[\Big(VIEW_{hybrid}(x,y), OUT_{hybrid}(x,y)\Big)
= (\perp, 0)\Big| l<i \bigwedge i \leq j \Big].$}

\noindent {\bf Subcase 1(b):} $l=i$. In this case, $P_1$ obtains the correct output i.e. $0$ and $P_2$ outputs $0$. In {\em ideal model}, when $P_1$ aborts in round $l=i$, trusted party has already sent $0$ to $P_1$ and $f_2(x_i,y_j)$ to $P_2$. As $i \leq j$, $f_2(x_i,y_j)= 0$ (Equation~\ref{eqfw}). Hence, 

{\scriptsize
$\Pr\Big[\Big(VIEW_{ideal}(x,y), OUT_{ideal}(x,y)\Big)
 = (0, 0)\Big|l=i \bigwedge i \leq j \Big]$\\

\noindent $= \Pr\Big[\Big(VIEW_{hybrid}(x,y), OUT_{hybrid}(x,y)\Big)
= (0, 0)\Big| l=i \bigwedge i \leq j \Big]$.}

\noindent {\bf Subcase 1(c):} $l>i$. Here two cases can arise. i) $i< l \leq j$, in this case, $P_1$ obtains correct output and $P_2$ outputs $0$. ii) $i=j<l$, in this case, both $P_1$ and $P_2$ have already obtained $0$.  In {\em ideal model}, if $P_1$ aborts in round $l>i$, $P_1$ has already got its output value whereas trusted party sends $f_2(x_l,y_j)$ to $P_2$. When $i < l \leq j$, then $f_2(x_l,y_j)=0$ (Equation~\ref{eqfw}) whereas for $i=j$, $P_2$ has already got the correct output i.e $0$. Hence, 

{\scriptsize
 $\Pr\Big[\Big(VIEW_{ideal}(x,y), OUT_{ideal}(x,y)\Big)
 = (0, 0)\Big|l>i \bigwedge i \leq j \Big]$\\

\noindent $=\Pr\Big[\Big(VIEW_{hybrid}(x,y), OUT_{hybrid}(x,y)\Big)
 = (0, 0)\Big| l>i \bigwedge i \leq j \Big]$.}

\noindent {\bf Case 2: $i > j$}.\\
\noindent {\bf Subcase 2(a):} $l\leq j$. In this case, $P_1$ outputs $\perp$ and $P_2$ outputs $0$. In {\em ideal model}, if $P_1$ aborts in round $l\leq j$, $P_1$ outputs $\perp$ and trusted party sends $f_2(x_l,y_j)$ to $P_2$. As $l\leq j$, $f_2(x_l,y_j)= 0$ (Equation~\ref{eqfw}). Hence, 

{\scriptsize
$\Pr\Big[\Big(VIEW_{ideal}(x,y), OUT_{ideal}(x,y)\Big)
 = (\perp, 0)\Big|l\leq j \bigwedge i > j \Big]$\\

\noindent $=\Pr\Big[\Big(VIEW_{hybrid}(x,y), OUT_{hybrid}(x,y)\Big)
 = (\perp, 0)\Big| l\leq j \bigwedge i > j \Big]$.}

\noindent {\bf Subcase 2(b):} $j<l<i$. In this case, $P_1$ obtains $\perp$ and $P_2$ gets the correct output i.e. $1$. In {\em ideal model}, if $P_1$ aborts in round $j<l<i$, $P_1$ is given $\perp$ whereas the trusted party sends $f_2(x_l,y_j)$ to $P_2$. As $j<l$, then $f_2(x_l,y_j)=1$  (Equation~\ref{eqfw}). Hence, 

{\scriptsize
 $\Pr\Big[\Big(VIEW_{ideal}(x,y), OUT_{ideal}(x,y)\Big)
 = (\perp, 1)\Big|j<l<i \bigwedge i > j \Big]$\\

\noindent $=\Pr\Big[\Big(VIEW_{hybrid}(x,y), OUT_{hybrid}(x,y)\Big)
 = (\perp, 1)\Big| j<l<i \bigwedge i > j \Big]$.}

\noindent {\bf Subcase 2(c):}\label{2d} $j<l=i$. In this case, $P_1$ and $P_2$ both obtain the correct output i.e. $1$. In {\em ideal model}, if $P_1$ aborts in round $j<l=i$, $P_1$ is given $1$ whereas the trusted party sends $f_2(x_l,y_j)$ to $P_2$. As $j<l=i$, then $f_2(x_l,y_j)= f_2(x_i,y_j)=1$  (Equation~\ref{eqfw}). Hence, 

{\scriptsize
 $\Pr\Big[\Big(VIEW_{ideal}(x,y), OUT_{ideal}(x,y)\Big)
= (1, 1)\Big|j<l=i \bigwedge i > j \Big]$\\

\noindent $=\Pr\Big[\Big(VIEW_{hybrid}(x,y), OUT_{hybrid}(x,y)\Big)
 = (1, 1)\Big| j<l=i \bigwedge i > j \Big]$.}

When $i<l<m$, $P_1$ has no incentive to abort as in this case both $P_1$ and $P_2$ have already obtain their respective outputs.

Hence, from the above analysis, we can conclude that in the hybrid model, the adversary does no more harm than the ideal scenario. Thus our protocol achieve fairness in non-rational setting.

\begin{theorem}
In non-rational setting, the protocol $\bf\Pi^{QMP}_{Fair}$ achieves fairness.
\end{theorem}

 \section{Quantum Solution of Millionaires' Problem in Rational Setting} 
As discussed in Section~\ref{fair}, the definition of fairness changes in rational setting. Thus, we have to modify our protocol in Section~\ref{nonrational} for rational setting. 

\restylealgo{boxed}
\begin{algorithm}[htbp]
{\scriptsize
\textbf {Inputs:}\\
 The inputs of the $QRShareGen$ are $x_ i$ from $P_1$ and $y_j$ from $P_2$. If one of the received inputs is not in the correct domain, then both the parties are given $\perp$.\\
\textbf{Computation:}\\
Dealer does the following:
\begin{enumerate}
\item Chooses $r$ according to a geometric distribution $\mathcal{G}(\gamma)$ with parameter $\gamma$ and sets it as the revelation round, i.e., the round in which the value of $f(x_i,y_j) = (0,0)$ or $(1,1)$.
\item Chooses $d$ according to the geometrical distribution $\mathcal{G}(\gamma)$ and sets the total number of iterations as $m = r + d$.
\item For the revelation round, i.e., when $l = r$, dealer does the following:\\
\hspace*{2mm} (a) If $f(x_i, y_j) = (0,0)$, prepares two copies of $\ket{g_0} = \frac{1}{\sqrt{2}}(\ket {0}_1\ket{0}_2 + \ket {1}_1\ket{1}_2)$.  \\
\hspace*{2mm} (b) If $f(x_i, y_j) = (1,1)$, prepares two copies of $\ket{g_1} = \frac{1}{\sqrt{2}}(\ket {0}_1\ket{0}_2 - \ket {1}\ket{1}_2)$. 
\item For each $l \in \{1, \dots, m\}, l \neq r$, prepares two copies of  $\ket{g_2} = \frac{1}{\sqrt{2}}(\ket {0}_1\ket{1}_2 + \ket {1}_1\ket{0}_2)$.
\end{enumerate}
\textbf{Output:}\\
\begin{enumerate}
\item For  $l \in \{1, 2, \dots, m\}$ dealer 
prepares a list $list_w$ of shares for each party $P_w$, where $w \in \{1,2\}$ such that for each round each player is given two qubits, marked as $1$st and $2$nd, from two different entangled states.\\
\hspace*{2mm}(a) when $l=r$, $P_1$ is given $1$st halves from two copies of $\ket{g_0}$ or $\ket{g_1}$ depending on the value of $f(x_i,y_j)$ and $P_2$  is given the second halves from the same entangled states.\\ 
\hspace*{2mm}(b) for all other rounds, $P_1$ is given first halves from two different  $\ket{g_2}$ states, whereas $P_2$ is given the remaining halves from the same entangled states.\\
\hspace*{2mm}(c) The marking of the qubits for a round for each party is such that the $1$st (resp. $2$nd) qubit of $P_1$ is correlated with the $1$st (resp. $2$nd) qubit of $P_2$.\\
\hspace*{2mm}(d) each list contains $2m$ number of qubits.\\
\end{enumerate}
}
\caption{$QRShareGen$}
\label{algo_rqsharegen1}
\end{algorithm}
\restylealgo{boxed}
\begin{algorithm}[htbp]
{\scriptsize
\textbf{Inputs:}\\
 Each of $P_1$ and $P_2$ receives his corresponding list of shares.\\
\textbf{Computation:}\\
The players do the following.\\
\begin{enumerate}
\item Each round is subdivided into two sub-rounds.\\
\item In first sub-round, $P_2$ sends the first qubit of its list for that round to $P_1$. \\
\item In second sub-round, $P_1$ sends the second qubit of its list for that round to $P_2$. \\
\item After receiving the qubits from $P_2$, $P_1$ measures the two qubits in Bell basis.\\
\hspace*{2mm}(a) If in any round $l$ the measurement result is $\ket{g_3}$, $P_1$ aborts the protocol and reports forgery by $P_2$.\\
\hspace*{2mm}(b) Otherwise, if the measurement result is $\ket{g_0}$, concludes $f_1(x_i,y_j) = 0$. If it is $\ket{g_1}$, concludes $f_1(x_i,y_j) = 1$. If it is $\ket{g_2}$, concludes $\perp$. \\
\item  After receiving the qubits from $P_1$, $P_2$ measures the two qubits in Bell basis.\\
\hspace*{2mm}(a) If in any round $l$ the measurement result is $\ket{g_3}$, $P_2$ aborts the protocol and reports forgery by $P_1$.\\
\hspace*{2mm}(b) Otherwise, if the measurement result is $\ket{g_0}$, concludes $f_2(x_i,y_j) = 0$. If it is $\ket{g_1}$, concludes $f_2(x_i,y_j) = 1$. If it is $\ket{g_2}$, concludes $\perp$. \\
\end{enumerate}
\ \\
\textbf{Output:}\\
\begin{enumerate}
\item $P_1$ and $P_2$ obtain their outputs in iteration $r$.\\
\item If $P_2$ aborts in round $l$, i.e., does not send its share at that round and $l \leq r$, $P_1$ outputs $\perp$. If $l > r$, $P_1$ has already determined the output in iteration $r$. Thus it outputs that value. \\
\item  If $P_1$ aborts in round $l$, i.e., does not send its share at that round and $l \leq r$, $P_2$ outputs $\perp$. If $l > r$, $P_2$ has already determined the output in iteration $r$. Thus it outputs that value. \\
\end{enumerate}
}
\caption{$\bf\Pi^{QRMP}_{Fair}$}
\label{algo_rfairqcomp1}
\end{algorithm}

Our proposed protocol is described in Algorithm~\ref{algo_rqsharegen1} ($QRShareGen$) and Algorithm~\ref{algo_rfairqcomp1} ($\bf {\Pi^{QRMP}_{Fair}}$).
Here, some additional assumptions are required. For example, unlike the non-rational setting, both the players obtain the value of the function in a specific round called revelation round. We denote this by $r$. The position of $r$ in $m$ number of iteration is not revealed to the players and is chosen according to a geometric distribution $\mathcal{G}(\gamma)$, where the parameter $\gamma$ in turn depends on the utility values of the players. We here assume that $\gamma < \frac{U_w^{TT}-U_w^{NN}}{U_w^{TN}-U_w^{NN}}$.
Another assumption is that if any player chooses abort in any round $l$, we tell him whether this round is the revelation round or not~\cite{GroceK}. The term and condition of the game is that knowing whether the round is the revelation round or not, no player can revise his decision. Now we show that under this restriction and an assumption that $\gamma < \frac{U_w^{TT}-U_w^{NN}}{U_w^{TN}-U_w^{NN}}$, our protocol achieves fairness. 

\subsection{Security Analysis}
A Byzantine player can manipulate the share as well as can abort early. Firstly, we analyze the security issues assuming that the player manipulates the share. Secondly, we analyze fairness of the protocol considering early abort of the corrupted player.

\subsubsection{\bf{Security against Forgery}} 
Without loss of generality, let us assume that $P_1$ is corrupted by the Byzantine adversary and can manipulate the share (here, qubit). He can send an arbitrary qubit to $P_2$ or can swap the qubits of his list and can send an uncorrelated qubit to $P_2$. The analysis is almost same as Subsection~\ref{Byz_non}. The forgery is detected with significant probability.

If no cheating occurs, then in round $l \neq r$, $P_2$ will get $\ket{g_2} = \frac{1}{\sqrt{2}}\Big[\ket{0}_1\ket{1}_2 + \ket{1}_1\ket{0}_2\Big].$ Now, let us assume that $P_1$ sends an arbitrary share $\ket{\phi}=[\alpha\ket{0}_3+\beta\ket{1}_3]$ instead of the correct one. Thus, at round $l \neq r$, the state at the end of $P_2$ would be
\begin{eqnarray*}
 \rho_{2} & = & [tr_{P_1}(\rho)](\rho_{3}) \\
  & = & \frac{1}{{2}}\Big[\ket{1}_2\bra{1}\Big(|\alpha|^2\ket{0}_3\bra{0}
+ \alpha^*\beta\ket{1}_3\bra{0} \\
& & +\alpha\beta^*\ket{0}_3\bra{1}
+|\beta|^2\ket{1}_3\bra{1}\Big)\\
& & + \ket{0}_2\bra{0}\Big(|\alpha|^2\ket{0}_3\bra{0}
 + \alpha^*\beta\ket{1}_3\bra{0}\\
& & +\alpha\beta^*\ket{0}_3\bra{1}
+|\beta|^2\ket{1}_3\bra{1}\Big)\Big],
\end{eqnarray*}
where $\rho=\frac{1}{2}(\ket{0}_1\ket{1}_2 + \ket{1}_1\ket{0}_2)(\bra{0}_1\bra{1}_2 + \bra{1}_1\bra{0}_2)$ and $\rho_3=[|\alpha|^2\ket{0}_3\bra{0} + \alpha^*\beta\ket{1}_3\bra{0}+\alpha\beta^*\ket{0}_3\bra{1}+|\beta|^2\ket{1}_3\bra{1}]$.

$P_2$ will measure qubit $2$ and  qubit $3$. After measurement, $P_2$ will get either $\ket{g_0}$ or $\ket{g_1}$ or $\ket{g_2}$ or $\ket{g_3}$ with probability $\frac{1}{4}$ instead of $\ket{g_2}$ only (see Section~\ref{Byz_non}). As $P_2$ has no idea about the position of the revelation round, when he gets $\ket{g_0}$ or $\ket{g_1}$, he conclude that this is the revelation round. When he gets $\ket{g_2}$, he concludes that $l\neq r$. Only if he gets $\ket{g_3}$, he immediately concludes that $P_1$ is cheating. The success probability of detecting such cheating for a round $l \neq r$ is $\frac{1}{4}$. 

Similarly, when $l=r$, if $P_1$ does not cheat, $P_2$ would get either $\ket{g_0}$ or $\ket{g_1}$ depending on the value of $f_2(x_i,y_j)$. However, if $P_1$ cheats, when measured, $P_2$ will get any one Bell state. In case of $\ket{g_0}$ and $\ket{g_1}$, he can not detect the cheating because he does not know the value of $f_2(x_i,y_j)$ a priori. Again, if he gets $\ket{g_2}$ he also does not detect the cheating, as in this case, he concludes that $l \neq r$. $P_2$ can immediately detect the cheating with certainty if and only if he gets $\ket{g_3}$. Thus, the success probability of detecting the cheating when $l=r$ is also $\frac{1}{4}$. As, $P_1$ have no idea about the position of $r$, the average success probability of $P_2$ to detect such cheating for a round $l$ is 
\begin{eqnarray*}
\frac{1}{4}\Pr(l \neq r)  +  \frac{1}{4}\Pr(l=r)
 =  \frac{1}{4}[ \Pr(l<r) + \Pr(l>r)]\\
 + \frac{1}{4}\Pr(l=r).
\end{eqnarray*}
We do not bother about $ \Pr(l>r)$ because, $P_2$ should have no incentive to detect the cheating when $l>r$, as he has already got $f_2(x_i,y_j)$ in round $r$. Thus the total success probability of detecting the cheating is
\begin{equation*}
\frac{1}{4}\Pr(l<r) +  \frac{1}{4}\Pr(l=r)
 = \frac{1}{4}(1-\gamma) + \frac{1}{4}\gamma
 = \frac{1}{4}
\end{equation*}

According to our protocol, if $P_2$ detects cheating, he will abort the protocol. Thus, if $\frac{1}{4} < U_1^{TT}$, $P_1$ has no incentive to forge in any round.
Same thing happens if we assume that $P_2$ is corrupted by the adversary.
\begin{theorem}
In rational setting, the success probability of an honest player to detect cheating in an arbitrary round $l$ by a player who is corrupted by a Byzantine adversary is $\frac{1}{4}.$
\end{theorem}
\begin{theorem}
\label{forgery}
In rational setting, if $\frac{1}{4} < U_w^{TT}$, where $w\in\{1,2\}$, no player has any incentive to forge in a round $l$.
\end{theorem}
Same conclusion can be drawn if we assume that $P_1$ swaps the quits of his list and sends an uncorrelated qubit to $P_2$. 

\subsubsection{\bf{Fairness against Early Abort}}
We have mentioned earlier that a player who is corrupted by a Byzantine adversary can abort early. As $P_1$ is always computing its output first followed by $P_2$, the aborting of $P_1$ plays an important role to achieve the fairness of the protocol. The early abort of $P_2$ will terminate the protocol up to that round in which $P_2$ aborts. In that case, either both get the correct outputs or none gets the correct outputs. Thus, early abort of $P_2$ does not affect the fairness condition. We now concentrate on the early abort of $P_1$.

Let us assume that $P_1$ aborts in round $l$. According to our protocol if $P_1$ declares early abort, we will tell whether the round is the revelation round or not. Knowing that $P_1$ can not revise his decision. If $l<r$, $P_1$ gets $\ket{g_2}$, whereas $P_2$ outputs $\perp$. That means in this case, the utility of $P_1$ is $U_1^{NN}$ (no one gets the output). If $P_1$ aborts in round $l=r$, $P_1$ gets either $\ket{g_0}$ or $\ket{g_1}$ depending on the value of $f_1(x_i,y_j)$ and $P_2$ outputs $\perp$. In this case, the utility of $P_1$ is $U_1^{TN}$ ($P_1$ gets the output and $P_2$ does not). $P_1$ should have no incentive to abort in round $l>r$, as in this case $P_1$ and $P_2$ both have already obtained the value of the function in iteration $r$. Thus, the expected utility of $P_1$ is
\begin{equation*}
U_1^{NN}\Pr(l<r) + U_1^{TN}\Pr(l=r)
 = U_1^{NN}(1-\gamma) + U_1^{TN}\gamma
\end{equation*}
According to our assumption that $\gamma < \frac{U_w^{TT}-U_w^{NN}}{U_w^{TN}-U_w^{NN}}$, we can write $U_1^{NN}(1-\gamma) + U_1^{TN}\gamma < U_1^{TT}$. Hence, $P_1$ should have no incentive to abort early in the protocol and the protocol achieves fairness.

\begin{theorem}
\label{fairness}
In rational setting, provided $\mathcal{R}_1$ (Section~\ref{pre}), $0<\gamma<1$ and $U_w^{TN}+(1-\gamma)U_w^{NN}<U_w^{TT}$ for all $w \in \{1,2\}$, the protocol $\bf\Pi^{QRMP}_{Fair}$ achieves fairness.
\end{theorem}
Now we are in a position to prove strict Nash equilibrium for our protocol $\bf\Pi^{QRMP}_{Fair}$. 
\begin{theorem}
In rational setting, provided $\frac{1}{4} < U_w^{TT}$, $\mathcal{R}_1$ (Section~\ref{pre}), $0<\gamma<1$ and $U_w^{TN}+(1-\gamma)U_w^{NN}<U_w^{TT}$ for all $w \in \{1,2\}$, the protocol $\bf\Pi^{QRMP}_{Fair}$ achieves strict Nash equilibrium.
\end{theorem}
\begin{proof}
In Theorem~\ref{forgery}, it has been shown that if $\frac{1}{4} < U_w^{TT}$, where $w\in\{1,2\}$, no player has any incentive to cheat. It will be better for him to follow the suggested strategy as by cheating he can not increase his payoff. Further in Theorem~\ref{fairness}, we proved that provided $\mathcal{R}_1$ (Section~\ref{pre}), $0<\gamma<1$ and $U_w^{TN}+(1-\gamma)U_w^{NN}<U_w^{TT}$ for all $w \in \{1,2\}$, no player has any incentive to abort early. In this case also, deviation from the suggested strategy does not help him to gain more payoff. 
In other word, we have $u_w (\sigma'_w,\overrightarrow{\sigma}_{-w} ) < u_w (\overrightarrow{\sigma})$ for any player $P_w$, $w \in \{1,2\}$ and hence
the player $P_w$ always follows the suggested strategy. \end{proof}
  
\section{Secure Two-Party Computation involving Embedded XOR}
In this section, we first describe the embedded XOR problem proposed by Gordon et al.~\cite{Gordon}. Let us denote two players by $P_1$ and $P_2$.
Player $P_1$ is given an ordered list $\{x_1,x_2,x_3\}$ and $P_2$ is given an ordered list $\{y_1,y_2\}$.  $P_1$ randomly chooses the input from the ordered list and sent to the dealer. $P_2$ also randomly chooses the input from his list and delivers to the dealer. Dealer calculates the function.
For convenience, we here recall the table for $f$ given in~\cite{Gordon}.
\begin{center}
  \begin{tabular}{| l  | c |  r|}\hline
    & $y_1$ & $y_2$ \\ 
    \hline
  $x_1$  &  0 & 1 \\ \hline
  $x_2$  & 1 & 0  \\ \hline
  $x_3$  &1 & 1  \\ \hline
  \end{tabular}
\end{center}
The function can be described as
\begin{equation}
\label{eqfw2}
 f_w(x, y) =
  \begin{cases}
   1 & \text{if } i \neq j; \\
   0 & \text{if } i = j.
  \end{cases}
\end{equation}
where,  $x$ and $y$ denote the inputs from $P_1$ and $P_2$ respectively and $w \in \{1,2\}$
The protocol proceeds in a series of $M$ iterations, where $M= \omega (\log \lambda)$, $\lambda$ is the security parameter.
The dealer chooses the revelation round $r$ according to geometric distribution with parameter $\gamma$.
The dealer then creates two sequences $\{a_l\}$ and $\{b_l\}$,
$l = 1, 2, \ldots, M$, as follows. 
\begin{eqnarray*}
\mbox{ For } l \geq r, & a_l = f_1(x,y) = b_l = f_2(x, y).\\
\mbox{ For } l < r, & a_l = f_1(x,\hat{y}), \indent b_l = f_2(\hat{x},y),
\end{eqnarray*}
where $\hat{x}$ (or $\hat{y}$) is a random value of $x$ (or $y$) chosen by the dealer.

Next, the dealer splits the secret $a_l$ into the shares $a^1_l$ and $a^2_l$,
and the secret $b_l$ into the shares $b^1_l$ and $b^2_l$, so that
$a_l = a^1_l \oplus a^2_l$ and 
$b_l = b^1_l \oplus b^2_l$, 
and gives
the shares $\{(a^1_l$, $b^1_l)\}$ to $P_1$ and
the shares $\{(a^2_l$, $b^2_l)\}$ to $P_2$.
In each round $l$, $P_2$ sends $a^2_l$ to $P_1$, who, in turn
sends $b^1_l$ to $P_2$. $P_1$ (res. $P_2$) learns the output value $f_1(x, y)$ (res. $f_2(x,y)$) in iteration $r$.  Here we assume that the dealer who will distribute the shares is honest and can compute the function described in Equation~\eqref{eqfw2}.

The algorithms in the Byzantine setting are the same as those in the fail-stop setting except some additional steps. In Byzantine setting, the shares are signed by the dealer. Exploiting MAC signature we can resist a player to send a false share.

\section{Quantum Protocol for Embedded XOR in Non-Rational Setting} 
We suitably modify the classical protocol by Gordon et al. to propose a quantum 
solution of the embedded XOR problem. As in the quantum protocol to solve the millionaires' problem, here also we exploit entangled states to obtain the security.

Now we describe the protocol. Let
$P_1$ is given an ordered list $\{x_1,x_2,x_3\}$ and $P_2$ is given an ordered list $\{y_1,y_2\}$. $P_1$ randomly chooses an input $x$ from his ordered list and sends to the dealer. Similarly, $P_2$ also chooses an input $y$ randomly from his ordered list and sends to the dealer. Dealer computes the function and  creates two sequences $\{a_l\}$ and $\{b_l\}$,
$l = 1, 2, \ldots, m$,  where $m$ is the total number of the round in such a way that

For $l \geq r$, $a_l = f_1(x,y) = b_l = f_2(x, y)$ and

For$l < r$, $a_l = f_1(x,\hat{y}), \indent b_l = f_2(\hat{x},y)$,

\noindent where $\hat{x}$ (or $\hat{y}$) is a random value of $x$ (or $y$) chosen by the dealer.
 In quantum domain, the two sequences $\{a_l\}$ and $\{b_l\}$ are distributed by exploiting the qubits of entangled states. The mechanism is described in Algorithm~\ref{algo_qesharegen} ($QEShareGen$)
and Algorithm~\ref{algo_efairqcomp} ($\bf {\Pi^{QEP}_{Fair}}$).

\subsection{Security Analysis} 
In this subsection we discuss the security issues against a Byzantine adversary. First, we analyze the sensitivity of our protocol to detect a cheating by a Byzantine player. Then we analyze the fairness issue when a player aborts early.

\subsubsection{\bf{Security against Forgery}}
Without loss of generality, we assume that $P_1$ is corrupted by the Byzantine adversary and tries to manipulate the qubits. According to our protocol, in any round $l \leq r$, if $P_1$ does not cheat, $P_2$ will measure either $\ket{g_0}$ or $\ket{g_1}$ depending on the value of $b_l$. However, when $P_1$ cheats, the case will be different. Let us assume that in round $l \leq r$, $P_1$ sends an arbitrary qubit $ \ket{\phi}=\alpha\ket{0}_3+\beta\ket{1}_3$ to $P_2$. {\em {Here, we assume that if $P_1$ would not cheat at the round $l$, $P_2$ would receive $\ket{g_0}$. Same thing happen if we assume that $P_2$ will receive $\ket{g_1}$.}}
Thus the final state at the end of $P_2$ would be 
\begin{eqnarray*}
 \rho_{2} & = & [tr_{P_1}(\rho)](\rho_{3})\\
  & = & \frac{1}{{2}}\Big[\ket{1}_2\bra{1}\Big(|\alpha|^2\ket{0}_3\bra{0}
+ \alpha^*\beta\ket{1}_3\bra{0}\\
& & +\alpha\beta^*\ket{0}_3\bra{1}+|\beta|^2\ket{1}_3\bra{1}\Big)\\
& & + \ket{0}_2\bra{0}
\Big(|\alpha|^2\ket{0}_3\bra{0} + \alpha^*\beta\ket{1}_3\bra{0}\\
& & +\alpha\beta^*\ket{0}_3\bra{1} +|\beta|^2\ket{1}_3\bra{1}\Big)\Big],
\end{eqnarray*}
where $\rho=\frac{1}{2}(\ket{0}_1\ket{1}_2 + \ket{1}_1\ket{0}_2)(\bra{0}_1\bra{1}_2 + \bra{1}_1\bra{0}_2)$ and $\rho_3=[|\alpha|^2\ket{0}_3\bra{0} + \alpha^*\beta\ket{1}_3\bra{0}+\alpha\beta^*\ket{0}_3\bra{1}+|\beta|^2\ket{1}_3\bra{1}]$.

$P_2$ will measure qubit $2$ and  qubit $3$. Thus, after measurement, $P_2$ will get either $\ket{g_0}$ or $\ket{g_1}$ or $\ket{g_2}$ or $\ket{g_3}$ with probability $\frac{1}{4}$ instead of $\ket{g_0}$ only (see Section~\ref{Byz_non}). As in round $l \leq r$, $P_2$ will measure either $\ket{g_0}$ or $\ket{g_1}$ without any cheating, when he gets $\ket{g_0}$ or $\ket{g_1}$, he can not detect cheating. If he gets $\ket{g_2}$ or $\ket{g_3}$, he immediately concludes that $P_1$ is cheating. Thus, the success probability of detecting such cheating for any round $l \leq r$ is $\frac{1}{2}$. After the revelation round, $P_2$ has no incentive to detect cheating as $P_2$ has already got the correct output. Thus, we can write the expected success probability of $P_2$ to detect cheating by $P_1$ is 
\begin{equation*}
\frac{1}{2}\Pr(l<r) +  \frac{1}{2}\Pr(l=r)
 = \frac{1}{2}(1-\gamma) + \frac{1}{2}\gamma
 = \frac{1}{2}.
\end{equation*}
The same situation arises when we assume that $P_2$ is cheating.
\begin{theorem}
In non-rational setting, in an arbitrary round $l \leq r$, the success probability of an honest player to detect cheating by a player who is corrupted by a Byzantine adversary is $\frac{1}{2}$.
\end{theorem}
The swapping of the qubits in a round i.e interchanging the position of the $1$st and $2$nd qubits can be analyzed in the same manner. 

\subsubsection{\bf{Fairness against Early Abort}} 
In this subsection we will show how the fairness condition is maintained when a player corrupted by a Byzantine adversary aborts the protocol prematurely.
Let us assume that $P_1$ aborts in round $l$. As $P_1$ is always computing its output first followed by $P_2$, the aborting of $P_1$ plays an important role to achieve the fairness of the protocol. The early abort of $P_2$ will terminate the protocol up to that round in which $P_2$ aborts. In that case, either both get the correct value or none gets the correct value. Thus, early abort of $P_2$ does not affect the fairness condition. 

According to our protocol, if $P_1$ aborts in a round $l<r$, $P_2$ outputs $b_{l-1}=f_2(\hat{x},y)$. In this case $P_1$ outputs $a_l$ which is equal to $f_1(x,\hat{y})$. In {\em ideal model}, for $l<r$ the trusted party sends $f_1(x,\hat{y})$ to $P_1$ and $f_2(\hat{x},y)$ to $P_2$. Thus, if $P_1$ aborts in round $l<r$, $P_1$ gets $f_1(x,\hat{y})$ and the trusted party sends $f_2(\hat{x},y)$ to $P_2$. Thus  for $l<r$, we get

{\scriptsize
 $\Pr\Big[\Big(VIEW_{ideal}(x,y), OUT_{ideal}(x,y)\Big)
= (f(x,\hat{y}), f(\hat{x},y))|l<r\Big]$\\

\noindent $= \Pr\Big[\Big(VIEW_{hybrid}(x,y), OUT_{hybrid}(x,y)\Big) 
= (f(x,\hat{y}), f(\hat{x},y))|l<r\Big]$.}

When $l=r$, $P_1$ has already got the correct output whereas $P_2$ outputs $b_{r-1}=f_2(\hat{x}, y)$. In {\em ideal model}, when $l=r$, the trusted party sends $f_1(x,y)$ (res. $f_2(x,y)$) to $P_1$ (res. $P_2$). The following analysis shows how fairness is maintained in this case.

Here, we first recall the table for embedded XOR. We get that in case of $P_1$, $\Pr[f_1(x_1,\hat{y})=0]=\Pr[f_1(x_2,\hat{y})=0]=\Pr[y \in \{y_1,y_2\}]=\frac{1}{2}$ and $\Pr[f_1(x_1,\hat{y})=1]=\Pr[f_1(x_2,\hat{y})=1]=\Pr[y \in \{y_1,y_2\}]=\frac{1}{2}$ whereas $\Pr[f_1(x_3,\hat{y})=0]=0$ and $\Pr[f_1(x_3,\hat{y})=1]=1$. In case of $P_2$, $\Pr[f_2(\hat{x},y)=0]=\Pr[x=x_1]=\frac{1}{3}$ and $\Pr[f_2(\hat{x},y)=1]=\Pr[x \in \{x_2,x_3\}]=\frac{2}{3}$. Thus, we can write the followings.
{\scriptsize
\begin{eqnarray*}
\Pr\Big[\Big(VIEW_{ideal}(x,y), OUT_{ideal}(x,y)\Big)
 = (0, 0)\Big|l = r\Big] & = & \frac{1}{3}. \frac{1}{3},\\
\Pr\Big[\Big(VIEW_{ideal}(x,y), OUT_{ideal}(x,y)\Big)
 = (0, 1)\Big|l = r \Big] & = & \frac{1}{3}. \frac{2}{3},\\
 \Pr\Big[\Big(VIEW_{ideal}(x,y), OUT_{ideal}(x,y)\Big)
 = (1, 0)\Big|l = r \Big] & = & \frac{2}{3}. \frac{1}{3},\\
 \Pr\Big[\Big(VIEW_{ideal}(x,y), OUT_{ideal}(x,y)\Big) 
= (1, 1)\Big|l = r  \Big] & = & \frac{2}{3}. \frac{2}{3}.
\end{eqnarray*}}
Similarly, in hybrid world,
{\scriptsize
\begin{eqnarray*}
 \Pr\Big[\Big(VIEW_{hybrid}(x,y), OUT_{hybrid}(x,y)\Big)
 = (0, 0)\Big|l = r \Big] & =  & \frac{1}{3}. \frac{1}{3},\\
 \Pr\Big[\Big(VIEW_{hybrid}(x,y), OUT_{hybrid}(x,y)\Big)
 = (0, 1)\Big|l = r \Big] & = & \frac{1}{3}. \frac{2}{3},\\
 \Pr\Big[\Big(VIEW_{hybrid}(x,y), OUT_{hybrid}(x,y)\Big)
= (1, 0)\Big|l = r \Big] & = & \frac{2}{3}. \frac{1}{3},\\
 \Pr\Big[\Big(VIEW_{hybrid}(x,y), OUT_{hybrid}(x,y)\Big)
 = (1, 1)\Big|l = r \Big] & = & \frac{2}{3}. \frac{2}{3}.
\end{eqnarray*}}
Above probability calculations show that when $l=r$ the adversary does not do more harm in hybrid world than that he can do in the ideal world. Thus, our protocol achieves fairness.

Fairness is obvious if we consider the abort of $P_1$ at a round $l>r$, as in this situation, both in ideal world and in hybrid world, $P_1$ as well as $P_2$ obtain the correct output in iteration $r$.
\begin{theorem}
In non-rational setting, in an arbitrary round $l$, the protocol $\bf\Pi^{QEP}_{Fair}$ achieves  fairness considering early abort of a corrupted player.
\end{theorem}

\restylealgo{boxed}
\begin{algorithm}[htbp]
{\scriptsize
\textbf {Inputs:}\\
 The inputs of the $QEShareGen$ are $x$ from $P_1$ and $y$ from $P_2$. If one of the received inputs is not in the correct domain, then both the parties are given $\perp$.\\
\textbf{Computation:}\\
Dealer does the following:\\
\begin{enumerate}
\item Chooses $r$ according to a geometric distribution $\mathcal{G}(\gamma)$ with parameter $\gamma$ and sets it as the revelation round, i.e., the round in which the value of $f(x,y)=(0,0)$ or $(1,1)$.
\item Chooses $d$ according to the geometrical distribution $\mathcal{G}(\gamma)$ and sets the total number of iterations as $m = r + d$.
\item \textbf {For $P_1$ }\\
\hspace*{2mm} (A) For $l< r$, in each round, the dealer calculates $a_{l}=f_1(x,\hat{y})$, where $\hat{y}$ is a random variable chosen by the dealer from the ordered list of $P_2$.\\
\hspace*{2mm} (i) If $a_{l}=0$, prepares $\ket{g_0}$. We call it $\ket{g'_0}_{<r}$.\\
\hspace*{2mm}(ii)  If $a_{l}=1$, prepares $\ket{g_1}$. We call it $\ket{g'_1}_{<r}$.\\
\hspace*{2mm} (B) For $l \geq r$, the dealer calculates $a_l = f_1(x,y)$.\\
\hspace*{2mm} (i) If $a_l=0$, prepares $\ket{g_0}$. We mark it as $\ket{g'_0}_{\geq r}$.\\
\hspace*{2mm} (ii) If $a_l=1$, prepares $\ket{g_1}$.  We mark it as $\ket{g'_1}_{\geq r}$. \\
\textbf{For $P_2$}\\
\hspace*{2mm}(A) For $l< r$,  in each round, the dealer calculates $b_{l}=f_2(\hat{x},y)$,  where $\hat{x}$ is a random variable chosen by the dealer from the ordered list of $P_1$.\\ 
\hspace*{2mm} (i) If $b_{l}=0$, prepares $\ket{g_0}$. We call it $\ket{g''_0}_{<r}$.\\
\hspace*{2mm}(ii)  If $b_{l}=1$, prepares $\ket{g_1}$. We call it $\ket{g''_1}_{<r}$. \\ 
\hspace*{2mm}(B) For $l \geq r$, the dealer calculates $b_l = f_2(x,y)$.\\
\hspace*{2mm} (i) If $b_l=0$, prepares $\ket{g_0}$. We mark it as $\ket{g''_0}_{\geq r}$.\\
\hspace*{2mm} (ii) If $b_l=1$, prepares $\ket{g_1}$. We mark it as $\ket{g''_1}_{\geq r}$.
\end{enumerate}
\textbf{Output:}\\
\begin{enumerate}
\item For  $l \in \{1, 2, \dots, m\}$ dealer 
prepares a list $list_w$ of shares for each party $P_w$, where $w \in \{1,2\}$ such that:\\
\hspace*{2mm} (a) For $ l<r$, in each round $P_1$ is given the first half from $\ket{g'_0}_{<r}$ or  $\ket{g'_1}_{<r}$ depending on the value of $a_{l}$. This qubit is marked as $1$st qubit for that round. $P_1$ is also given the first half from $\ket{g''_0}_{<r}$ or  $\ket{g''_1}_{<r}$ depending on the value of $b_{l}$. This qubit is marked as $2$nd qubit for that round. \\
\hspace*{2mm} (b) For $ l \geq  r$, in each round $P_1$ is given the first half from $\ket{g'_0}_{\geq r}$  or  $\ket{g'_1}_{\geq r}$ depending on the value of $a_l$. This qubit is marked as $1$st qubit for that round. $P_1$ is also given the first half from $\ket{g''_0}_{\geq r}$ or $\ket{g''_1}_{\geq r}$ depending on the value of $b_l$. This qubit is marked as $2$nd qubit for that round.\\ 

\hspace*{2mm} (c) Similarly, for $ l<r$, in each round $P_2$ is given the remaining half from $\ket{g'_0}_{<r}$ or  $\ket{g'_1}_{<r}$ depending on the value of $a_{l}$. This qubit is marked as $1$st qubit for that round.  $P_2$ is also given the remaining half from $\ket{g''_0}_{<r}$ or  $\ket{g''_1}_{<r}$ depending on the value of $b_{l}$. This qubit is marked as $2$nd qubit for that round. \\
\hspace*{2mm} (d) For $ l \geq r$, in each round $P_2$ is given the remaining half from  $\ket{g'_0}_{\geq r}$ or  $\ket{g'_1}_{\geq r}$ depending on the value of $a_l$. This qubit is marked as $1$st qubit for that round. $P_2$ is also given the remaining half from $\ket{g''_0}_{\geq r}$ or  $\ket{g''_1}_{\geq r}$ depending on the value of $b_{l}$. This qubit is marked as $2$nd qubit for that round. 
\item Each list consists of $2m$ number of qubits.
\end{enumerate}
}
\caption{$QEShareGen$ }
\label{algo_qesharegen}
\end{algorithm}

\restylealgo{boxed}
\begin{algorithm}[htbp]
{\scriptsize
\textbf{Inputs:}\\
 Each of $P_1$, $P_2$ receives his corresponding list of shares.\\
\textbf{Computation:}\\
The players do the following.\\
\begin{enumerate}
\item Each round is subdivided into two sub-rounds.\\
\item In first sub-round, $P_2$ sends the first qubit of its list to $P_1$. \\
\item In second sub-round, $P_1$ sends the second qubit of its list to $P_2$. \\
\item After receiving the qubits from $P_2$, $P_1$ measures the two qubits in $Bell$ basis. If it will be $\ket{g_0}$, then concludes $a_l= 0$. If it will be $\ket{g_1}$, concludes $a_l = 1$. \\
\item  After receiving the qubits from $P_1$, $P_2$ measures the two qubits in $Bell$ basis. If it will be $\ket{g_0}$, then concludes $b_l = 0$. If it will be $\ket{g_1}$, concludes $b_l = 1$.\\
\item If in any round, any player $P_w$, measures $\ket{g_2}$ or $\ket{g_3}$, he immediately aborts the protocol and reports forgery by the other player. 
\end{enumerate}
\ \\
\textbf{Output:}\\
\begin{enumerate}
\item If $P_2$ aborts in round $l$, i.e., does not send its share at that round and $l \leq r$, $P_1$ outputs $a_{l-1}$. If $l > r$, $P_1$ has already determined the correct output in iteration $r$. Thus it outputs that value. \\
\item  If $P_1$ aborts in round $l$, i.e., does not send its share at that round and $l \leq r$, $P_2$ outputs $b_{l-1}$. If $l > r$, $P_2$ has already determined the correct output in iteration $r$. Thus it outputs that value. \\
\end{enumerate}
}
\caption{$\bf\Pi^{QEP}_{Fair}$}
\label{algo_efairqcomp}
\end{algorithm}

\section{Quantum Protocol for Embedded XOR in Rational Setting} In rational setting fairness means either everyone gets the correct output value or none gets it. Thus, in rational setting, we redefined the fairness condition (Section~\ref{pre}). It is immediate that when $P_1$ chooses $x=x_3$, he should have no incentive to continue the game, as in certainty, he knows that the output value is equal to $1$. In this situation, $P_2$ outputs $f_2(\hat{x},y)$ which may be $0$ with probability $\frac{1}{3}$ and may be $1$ with probability $\frac{2}{3}$. Thus, fairness condition in rational setting is violated. To mitigate the problem, we have to modify our protocol. In rational setting, we only modify {\em {step $2$ of the output portion of the protocol $\bf\Pi^{QEP}_{Fair}$}}. If $P_1$ aborts in any round $l \leq r$, instead of $b_{l-1}$, $P_2$ outputs $1$. Now, we will show how our new protocol $\bf\Pi^{QEP2}_{Fair}$ achieves fairness under some suitable choice of the parameters in the rational setting.

\subsection{Security Analysis} 
The security analysis against Byzantine adversary in rational setting is proceed exactly the same manner as the security analysis against Byzantine adversary in non-rational setting. We first analyze the cheating situation and then will discuss the fairness issue when a player aborts early.

\subsubsection{\bf{Security against Forgery}} 
This goes exactly the same way as it goes in non-rational setting.
\begin{theorem}
In rational setting, in an arbitrary round $l \leq r$, the success probability of an honest player to detect cheating by a player who is corrupted by a Byzantine adversary is $\frac{1}{2}$.
\end{theorem}
If $U_w^{TT}$, where $w \in \{1,2\}$, is greater than $\frac{1}{2}$, $P_w$ should have no incentive to cheat. Thus,
\begin{theorem}
\label{forgexor}
In rational setting, if $\frac{1}{2} < U_w^{TT}$, where $w\in\{1,2\}$, no player has any incentive to forge in a round $l$.
\end{theorem}

\subsubsection{\bf{Fairness against Early Abort}} 
The analysis against Byzantine adversary when he chooses early abort is analyzed in this subsection. We do not bother about the early abort of $P_2$, as early aborting of $P_2$ does not affect the fairness condition of the protocol.
\begin{center}
{\small{\bf{Early abort by $P_1$}}}
\end{center}
\label{p1earlyfair}
Now,  we discuss each case one by one.

\noindent {\bf {Case 1: $x = x_1$}}.
We have $\Pr( a_l = 0| x = x_1) = \Pr(\hat{y}=y_1) = \frac{1}{2}$ and  $\Pr(a_l = 1 | x = x_1) = \Pr(\hat{y} = y_2) = \frac{1}{2}$, for $l < r$. Note that
for $l=r$, $P_1$ will abort after receiving the exact value of $y$.
Hence, {in case of $y=y_1$, }
$$
\Pr( a_{r} = 0| (x_1,y_1)) = 1, 
\Pr(a_{r} = 1 | (x_1,y_1)) = 0$$ 
and in case of $y=y_2$, 
$$
\Pr( a_{r} = 0| (x_1,y_2)) = 0,
\Pr(a_{r} = 1 | (x_1,y_2)) = 1.$$

\noindent{\bf Subcase 1(a):} $y = y_1$.
Now, we have $\Pr(b_l = 0 | y = y_1) = 0$ and
$\Pr(b_l = 1 | y = y_1) = 1$.

The following table enumerates the different possibilities for $U_1$ when
$x=x_1$ and $y=y_1$.
\begin{center}
\begin{tabular}{|c|c|c|c|}
\hline
$(a_l, b_l)$ & $U_1$ & \multicolumn{2}{c|}{Probability}\\
\cline{3-4}
& & $l < r$ & $l=r$\\
\hline
(0,0) & $U_1^{TT}$ 
& $(1-\gamma)\cdot\frac{1}{2}\cdot 0 = 0$
& $\gamma\cdot 1\cdot 0 = 0$\\ 
\hline
(0,1) & $U_1^{TN}$ 
& $(1-\gamma)\cdot\frac{1}{2}\cdot 1 = (1-\gamma)\cdot \frac{1}{2}$
& $\gamma\cdot 1\cdot 1 = \gamma\cdot 1$\\ 
\hline
(1,0) & $U_1^{NT}$ 
& $(1-\gamma)\cdot\frac{1}{2}\cdot 0 = (1-\gamma)\cdot 0$
& $\gamma\cdot 0\cdot 0 = 0$\\ 
\hline
(1,1) & $U_1^{NN}$ 
& $(1-\gamma)\cdot\frac{1}{2}\cdot 1 = (1-\gamma)\cdot\frac{1}{2}$
& $\gamma\cdot 0\cdot 1 = 0$\\ 
\hline
\end{tabular}
\end{center}
Thus, the expected utility of $P_1$ in this case is
\begin{eqnarray*}
E[U_1 | (x_1, y_1)] & = & {(1-\gamma)}\Big[\frac{1}{2}U_1^{TN} + \frac{1}{2}U_1^{NN} \Big] + {\gamma}\Big[U_1^{TN} \Big]\\
& = & \frac{(1+\gamma)}{2}\Big(U_1^{TN}\Big) + \frac{(1-\gamma)}{2}\Big(U_1^{NN}\Big).
\end{eqnarray*}

\noindent{\bf Subcase 1(b):} $y = y_2$.
Now, we have $\Pr(b_l = 0 | y = y_2) = 0$ and
$\Pr(b_l = 1 | y = y_2) = 1$.

The following table enumerates the different possibilities for $U_1$ when
$x=x_1$ and $y=y_2$.
\begin{center}
\begin{tabular}{|c|c|c|c|}
\hline
$(a_l, b_l)$ & $U_1$ & \multicolumn{2}{c|}{Probability}\\
\cline{3-4}
& & $l < r$ & $l=r$\\
\hline
(0,0) & $U_1^{NN}$
& $(1-\gamma)\cdot\frac{1}{2}\cdot 0 = (1-\gamma)\cdot 0$
& $\gamma\cdot 0\cdot 0 = 0$\\
\hline
(0,1) & $U_1^{NT}$
& $(1-\gamma)\cdot\frac{1}{2}\cdot 1 = (1-\gamma)\cdot \frac{1}{2}$
& $\gamma\cdot 0\cdot 1= 0$\\
\hline
(1,0) & $U_1^{TN}$
& $(1-\gamma)\cdot\frac{1}{2}\cdot 0 = (1-\gamma)\cdot 0$
& $\gamma\cdot 1\cdot 0 = 0$\\
\hline
(1,1) & $U_1^{TT}$
& $(1-\gamma)\cdot\frac{1}{2}\cdot 1 = (1-\gamma)\cdot\frac{1}{2}$
& $\gamma\cdot 1\cdot 1 = \gamma$\\
\hline
\end{tabular}
\end{center}
Thus, the expected utility of $P_1$ in this case is
\begin{eqnarray*}
E[U_1 | (x_1, y_2)] & = & (1-\gamma)\Big(\frac{1}{2}U_1^{TT} + \frac{1}{2}U_1^{NT}\Big) + \gamma \Big(U_1^{TT}\Big)\\
& = & \frac{(1+\gamma)}{2}\Big(U_1^{TT}\Big) + \frac{(1-\gamma)}{2}\Big(U_1^{NT}\Big).
\end{eqnarray*}

\noindent Now, combining all two subcases, we get
\begin{eqnarray*}
& & E[U_1 | x_1]\\
& = & E[U_1 | (x_1,y_1)]\cdot \Pr(y=y_1)
+ E[U_1 | (x_1,y_2)]\cdot \Pr(y=y_2)\\
& = & 
\Big[\frac{(1+\gamma)}{2}\Big(U_1^{TN}\Big) + \frac{(1-\gamma)}{2}\Big(U_1^{NN}\Big)\Big]\cdot\frac{1}{2}\\
& & + \Big[\frac{(1+\gamma)}{2}\Big(U_1^{TT}\Big) + \frac{(1-\gamma)}{2}\Big(U_1^{NT}\Big)\Big]\cdot\frac{1}{2}\\
& = & \frac{(1+\gamma)}{4}\Big(U_1^{TN}+U_1^{TT}\Big)
 + \frac{(1-\gamma)}{4}\Big(U_1^{NN}+U_1^{NT}\Big).
\end{eqnarray*}
If the above expression is greater than or equal to $U_1^{TT}$, 
$P_1$ chooses abort. Thus, for fairness, we need to ensure that 
$U_1^{TT} > \frac{(1+\gamma)}{4}\Big(U_1^{TN}+U_1^{TT}\Big) + \frac{(1-\gamma)}{4}\Big(U_1^{NN}+U_1^{NT}\Big)$, i.e.,
\begin{equation}
\label{faircond2}
\gamma < \frac{3U_1^{TT} - U_1^{TN} - U_1^{NN} - U_1^{NT}}{U_1^{TN} + U_1^{TT}-U_1^{NN} - U_1^{NT}}.
\end{equation}

\noindent {\bf {Case 2: $x = x_2$}}.
The analysis is similar and we obtain the same expression for $E[U_1 | x_2]$.
More specifically, we have the following observation.

\noindent{\bf Subcase 2(a):} $y = y_1$. The analysis is exactly identical to
Subcase 1(b).

\noindent{\bf Subcase 2(b):} $y = y_2$. The analysis is exactly identical to
Subcase 1(a).

\noindent {\bf {Case 3: $x = x_3$}}.
When $x = x_3$, $P_1$ will abort as he knows the output with certainty. In this case, he needs no help from $P_2$ to compute the function. However, when $P_1$ chooses to abort, $P_2$  outputs $1$. Thus, for $x=x_3$, both get the correct output of the function. The utility for both the player is $U_w^{TT}$, $w\in \{1,2\}$. Hence, the fairness condition in rational setting is always maintained.

\subsubsection{\bf {Fairness Condition}}
From the above analysis, we can state the following result.
\begin{theorem}
\label{fairxor_unequal}
Provided $\mathcal R_{1}$ (Section~\ref{pre}),  
$(U_1^{TT}-U_1^{NN})+(U_1^{TT}-U_1^{NT}) > (U_1^{TN}-U_1^{TT})$,
and
$$0 < \gamma < \frac{3U_1^{TT}-U_1^{TN}-U_1^{NN}-U_1^{NT}}{U_1^{TN}+U_1^{TT}-U_1^{NN}-U_1^{NT}},$$
the protocol  $\bf \Pi^{CEP2}_{Fair}$ achieves fairness.
\end{theorem}
\begin{proof}
The proof follows from Equations~\eqref{faircond2}.
The additional condition 
\begin{equation}
\label{ucond}
(U_1^{TT}-U_1^{NN})+(U_1^{TT}-U_1^{NT}) > (U_1^{TN}-U_1^{TT})
\end{equation}
follows from the fact that for $\gamma$ to be meaningful, the numerator $3U_1^{TT}-U_1^{TN}-U_1^{NN}-U_1^{NT}$ must be $\geq 0$.
Further, from the condition $\gamma < \frac{3U_1^{TT}-U_1^{TN}-U_1^{NN}-U_1^{NT}}{U_1^{TN}+U_1^{TT}-U_1^{NN}-U_1^{NT}}$, it is easy to see that the
natural restriction $\gamma < 1$ always holds. 
\end{proof}
In Equation~\eqref{ucond}, all the three terms within the parentheses are 
non-negative according to $\mathcal R_{1}$.

\subsubsection{\bf {Strict Nash Equilibrium}} 
Combining the above results, we can state the following.
\begin{theorem}
Provided $\frac{1}{2} < U_w^{TT}$  for $w\in\{1,2\}$, $\mathcal R_{1}$ (Section~\ref{pre}),  
$(U_1^{TT}-U_1^{NN})+(U_1^{TT}-U_1^{NT}) > (U_1^{TN}-U_1^{TT})$,
and
$$0 < \gamma < \frac{3U_1^{TT}-U_1^{TN}-U_1^{NN}-U_1^{NT}}{U_1^{TN}+U_1^{TT}-U_1^{NN}-U_1^{NT}},$$
the protocol  $\bf \Pi^{CEP2}_{Fair}$ achieve strict Nash equilibrium.
\end{theorem}
\begin{proof}
From Theorem~\ref{forgexor}, we get that provided $\frac{1}{2} < U_w^{TT}$ for $w\in\{1,2\}$, no player has any incentive to cheat as he can not increase his payoff by cheating. In case of early abort, $P_2$ cannot maximize his utility, as early abort of $P_2$ will terminate the protocol and in that case either no one gets the correct output ($U_2^{NN}$) or both get the correct output ($U_2^{TT}$). So $P_2$ never achieves $U_2^{TN}$ by aborting early. However, it is $P_1$ who can achieve $U_1^{TN}$ by aborting early, as $P_1$ always computes the output first followed by $P_2$. But in Theorem~\ref{fairxor_unequal}, we proved that provided $\mathcal R_{1}$ (Section~\ref{pre}),  
$(U_1^{TT}-U_1^{NN})+(U_1^{TT}-U_1^{NT}) > (U_1^{TN}-U_1^{TT})$,
and $0 < \gamma < \frac{3U_1^{TT}-U_1^{TN}-U_1^{NN}-U_1^{NT}}{U_1^{TN}+U_1^{TT}-U_1^{NN}-U_1^{NT}},$ $P_1$ has no incentive to abort early. Thus, we can say that for every player $P_w$, $w\in \{1,2\}$, $u_w (\sigma'_w,\overrightarrow{\sigma}_{-w} ) < u_w (\overrightarrow{\sigma})$ holds and hence no one deviates from the
suggested strategy.
\end{proof}

\section{Conclusion and Future Work}
In 1997, Lo~\cite{Lo} showed the impossibility of secure two-party 
quantum computation of certain functions, when one of the parties is malicious. In this direction, we obtain a positive result for two types of functions. This does not contradict with the generalized impossibility results of~\cite{Ben} in broadcast channel model, since we show our results in non-simultaneous channel model. 

Further, for the first time, we introduce the idea of secure two-party quantum computation with rational players. When one moves from the non-rational domain to a rational one, the definition for fairness changes. Thus, we modify the protocols to achieve fairness in rational setting. In addition, we prove strict Nash equilibrium for our proposed protocols in rational setting.

We have shown that secure two-party quantum computation is possible for any function without an embedded XOR and for a particular function with an embedded XOR. Thus, it remains an open question whether secure two-party quantum computation is possible for any function with an embedded XOR. Moreover, generalization of
the two-party protocols to $n$-party scenario would be an interesting future work, particularly, in the non-simultaneous channel model.


\begin{thebibliography}{}
\bibitem{footnote}
Here, `two-sided' means that both Alice and Bob are able to compute the functionality in two-party setting.
\end{thebibliography}
\end{document}